\documentclass[journal, onecolumn]{IEEEtran}
\usepackage{graphicx} 
\usepackage{amsmath}
\usepackage{amssymb}
\usepackage{amsthm}

\newtheorem{theorem}{Theorem}
\newtheorem{remark}{Remark} 
\newtheorem{assumption}{Assumption}

\newtheorem{lemma}{Lemma}
\usepackage{multicol}
\usepackage{multirow}
\usepackage[english]{babel}   
\usepackage{color}
\usepackage{multicol}
\usepackage[T1]{fontenc}

\ifCLASSINFOpdf
\else
\fi
\begin{document}
%
\title{Adaptive-Robust Control of a Class of Uncertain Nonlinear Systems Utilizing Time-Delayed Input and Position Feedback }
%
%
%

\author{Spandan~Roy
        and~Indra~Narayan~Kar,~\IEEEmembership{Senior~Member,~IEEE}
\thanks{
}
\thanks{}}

%
%

\markboth{}%
{Shell \MakeLowercase{\textit{et al.}}: Bare Demo of IEEEtran.cls for Journals}
%



\maketitle

\begin{abstract}
In this paper, the tracking control problem of a class of Euler-Lagrange systems subjected to unknown uncertainties is addressed and an adaptive-robust control strategy, christened as Time-Delayed Adaptive Robust Control (TARC) is presented. The proposed control strategy approximates the unknown dynamics through time-delayed logic, and the switching logic provides robustness against the approximation error. The novel adaptation law for the switching gain, in contrast to the conventional adaptive-robust control methodologies, does not require either nominal modelling or predefined bounds of the uncertainties. Also, the proposed adaptive law circumvents the overestimation-underestimation problem of switching gain. The state derivatives in the proposed control law is estimated from past data of the state to alleviate the measurement error when state derivatives are not available directly. Moreover, a new stability notion for time-delayed control is proposed which in turn provides a selection criterion for controller gain and sampling interval. Experimental result of the proposed methodology using a nonholonomic wheeled mobile robot (WMR) is presented and improved tracking accuracy of the proposed control law is noted compared to time-delayed control and adaptive sliding mode control.
\end{abstract}


%
\IEEEpeerreviewmaketitle

\section{Introduction}
%
%
%
%
\subsection{Background and Motivation}\label{intro}
\IEEEPARstart{D}{esign} of an efficient controller for nonlinear systems subjected to parametric and nonparametric uncertainties has always been a challenging task. Among many other approaches, Adaptive control and Robust control are the two popular control strategies that researchers have extensively employed while dealing with uncertain nonlinear systems. In general, adaptive control uses predefined parameter adaptation laws and equivalence principle based control law which adjusts the parameters of the controller on the fly according to the pertaining uncertainties \cite{Ref:1}. However, this approach has poor transient performance and online calculation of the unknown system parameters and controller gains for complex systems is computationally intensive \cite{Ref:2}. Whereas, robust control aims at tackling the uncertainties of the system within an uncertainty bound defined a priori. It reduces computation complexity to a great extent for complex systems compared to adaptive control as exclusive online estimation of uncertain parameters is not required \cite{Ref:3}. 
However, nominal modelling of the uncertainties is necessary to decide upon their bounds, which is not always possible. Again, to increase the operating region of the controller, often higher uncertainty bounds are assumed. This in turn leads to problems like higher controller gain and consequent possibility of chattering for the switching law based robust controller like Sliding Mode Control (SMC). This in effect reduces controller accuracy \cite{Ref:4}. Higher order sliding mode \cite{Ref:5} can alleviate the chattering problem but prerequisite of uncertainty bound still exists. 
\par  Time-Delayed Control (TDC) is utilized in \cite{Ref:6-1} to implement state derivative feedback for enhancing stability margin of SISO linear time invariant (LTI) systems. In \cite{Ref:6}, \cite{Ref:6-2}, \cite{Ref:25}-\cite{Ref:26}, \cite{Ref:33}, \cite{Ref:34}, TDC is used to provide robustness against uncertainties. In this process, all the uncertain terms are represented by a single function which is then approximated using control input and state information of the immediate past time instant. The advantage of this robust control approach in uncertain systems is that it reduces the burden of tedious modelling of complex system to a great extent. In spite of this, the unattended approximation error, commonly termed as time-delayed error (TDE) causes detrimental effect to the performance of the closed system and its stability. In this front, a few work have been carried out to tackle TDE which includes internal model \cite{Ref:7}, gradient estimator \cite{Ref:8}, ideal velocity feedback \cite{Ref:9}, nonlinear damping \cite{Ref:10} and sliding mode based approach \cite{Ref:11}-\cite{Ref:12}.  The stability of the closed loop system \cite{Ref:7}-\cite{Ref:9}, \cite{Ref:25}-\cite{Ref:26}, depends on the boundedness of TDE as shown in \cite{Ref:6}. This method approximates the continuous time closed loop system in a discrete form without considering the effect of discretization error. Again, the stability criterion mentioned in \cite{Ref:6} restricts the allowable range of perturbation and thus limits controller working range. Stability of the system in \cite{Ref:11} is established in frequency domain, which makes the approach inapplicable to the nonlinear systems. Moreover, the controllers designed in \cite{Ref:10}, and \cite{Ref:12}, \cite{Ref:34} require nominal modelling and upper bound of the TDE respectively which is not always possible in practical circumstances. Also, to the best knowledge of the authors, controller design issues such as selection of controller gains and sampling interval to achieve efficient performance is still an open problem. In contrast to TDC, works reported in \cite{Ref:27}-\cite{Ref:29} use low pass filter to approximate the unknown uncertainties and disturbances. However, frequency range of system dynamics and external disturbances are required to determine the time constant of the filter. Furthermore, the order of the low pass filter needs to be adjusted according to order of the disturbance to maintain stability of the controller.
\par Considering the individual limitations of adaptive and robust control, recently global research is reoriented towards adaptive-robust control (ARC) where switching gain of the controller is adjusted online. The series of publications \cite{Ref:2}, \cite{Ref:13}-\cite{Ref:19} regarding ARC, estimates the uncertain terms online based on predefined projection function, but predefined bound on uncertainties is still a requirement. The work reported in \cite{Ref:20}, \cite{Ref:32} attempts to estimate the maximum uncertainty bound but the integral adaptive law makes the controller susceptible to very high switching gain and consequent chattering \cite{Ref:33}. The adaptive sliding mode control (ASMC) as presented in \cite{Ref:21}-\cite{Ref:22} proposed two laws for the switching gain to adapt itself online according to the incurred error. In the first adaptive law, the switching gain decreases or increases depending on a predefined threshold value. However, until the threshold value is achieved, the switching gain may still be increasing (resp. decreasing) even if tracking error decreases (resp. increases) and thus creates overestimation (resp. underestimation) problem of switching gain \cite{Ref:23}. Moreover, to decide the threshold value the maximum bound of the uncertainty is required. For the second adaptive law, the threshold value changes online according to switching gain. Yet, nominal model of the uncertainties is needed for defining the control law. This limits the adaptive nature of the control law and applicability of the controller. 
\subsection{Problem Definitions and Contributions}\label{sec 1.2}
In this paper, three specific related problems on TDC have been dealt with and the corresponding solutions to the same which are also the contributions of this paper are summarized below: 
\begin{itemize}
\item \textit{ \textbf{Problem 1:} The stability analysis of TDC, as provided in \cite{Ref:6}, \cite{Ref:25}-\cite{Ref:26}, \cite{Ref:33}, \cite{Ref:34}, approximates the continuous time system in discrete time domain without considering the effects of discretization error. Again, choice of the delay time and its relation with the controller gains is still an open problem.} 

  In this paper, a new stability analysis for TDC, based on the Lyapunov-Krasvoskii method, is provided in continuous time domain. Furthermore, through the proposed stability approach, a relation between the sampling interval and controller gain is established.  

\item \textit{\textbf{Problem 2:} 
The TDC reported in (\cite{Ref:6}, \cite{Ref:25}-\cite{Ref:26}, \cite{Ref:33}, \cite{Ref:34}), velocity and acceleration feedback are necessary to compute the control law. While in \cite{Ref:6-1}, only velocity feedback is required and acceleration term is approximated numerically using time delay. However, in many applications velocity and acceleration feedback are not available explicitly and numerical approximation of these terms invokes measurement error.}

 As a second contribution of this paper, Filtered Time-Delayed Control (F-TDC) control law is formulated where only position feedback is sufficient while velocity and acceleration terms are estimated using past and present position information to curb the effect measurement error. Stability analysis of the proposed F-TDC is provided which also maintains the relation between controller gains and sampling interval.  
\item \textit{ \textbf{Problem 3:} Robustness property against TDE is essential to achieve good tracking accuracy. The robust controllers reported in literature, either requires nominal model of the uncertainties (\cite{Ref:21}-\cite{Ref:22}) or its predefined bound (\cite{Ref:2}, \cite{Ref:10}, \cite{Ref:12}, \cite{Ref:13}-\cite{Ref:19}). So, it is required to devise a control law which would avoid any prior knowledge of the uncertainties while providing robustness against TDE.}

 Towards the last contribution of this article, an adaptive-robust control strategy, Time-Delayed Adaptive Robust Control (TARC) has been formulated for a class of uncertain Euler-Lagrange systems. The proposed control law approximates uncertainties by time-delayed logic and provides robustness against the TDE, arising from time-delayed logic based estimation, by switching control. The novel adaptive law, presented here, aims at overcoming the overestimation-underestimation problem of the switching gain without any prior knowledge of uncertainties. The proposed adaptive law provides flexibility to the control designer to select any suitable error function according to the application requirement while maintaining similar system stability notion. 
\end{itemize}
\par As a proof of concept, experimental validation of the proposed control methodology is provided using the "PIONEER-3" nonholonomic WMR in comparison to TDC \cite{Ref:6} and ASMC \cite{Ref:21}-\cite{Ref:22}.

\subsection{Organization}
The article is organized as follows: a new stability analysis of TDC along with its design issues is first discussed in Section II. This is followed by the proposed adaptive-robust control methodology and its detail analysis. Section III presents the experimental results of the proposed controller and its comparison with TDC and ASMC. Section IV concludes the entire work.

\subsection{Notations}
The following notations are assumed for the entirety of the paper: any variable $\mu$ delayed by an amount $h$ as $\mu(t-h)$, is denoted as $\mu_{h}$; $\lambda _{min}(\cdot)$ and $|| \cdot ||$ represent minimum eigen value and Euclidean norm of the argument respectively; $I$ represents identity matrix.
%
%
\section{Controller Design}\label{sec: 2}
\subsection{Time-Delayed Control: Revisited}\label{sec: 2.1}
In general, an Euler-Lagrange system with second order dynamics, devoid of any delay, can be written as,
\begin{equation}\label{sys}
M(q)\ddot{q}+N(q,\dot{q})=\tau(t),
\end{equation}
where, $q(t)\in\mathbb{R}^{n}$ is the system state, $\tau(t)\in\mathbb{R}^{n}$ is the control input, $M(q)\in\mathbb{R}^{n\times n}$  is the mass/inertia matrix and $N(q,\dot{q})\in\mathbb{R}^{n}$ denotes combination of other system dynamics terms based on system properties. In practice, it can be assumed that unmodelled dynamics and disturbances is subsumed by $N$. The control input is defined to be,
\begin{equation}\label{input}
\tau=\hat{M}u+\hat{N},
\end{equation}
where, $u$ is the auxiliary control input, $\hat{M}$ and $\hat{N}$ are the nominal values of $M$ and $N$ respectively. To reduce the modelling effort of the complex systems, $\hat{N}$ can be approximated from the input-output data of previous instances using the time-delayed logic (\cite{Ref:6}, \cite{Ref:25}-\cite{Ref:26}) and the system definition (\ref{sys}) as,

\begin{equation}\label{approx}
\hat{N}( q,\dot{q} )\cong N(q_h,\dot{q}_h)=\tau_h-\hat{M}(q_h)\ddot{q}_h,
\end{equation}
where, $h>0$ is a fixed small delay time. Substituting (\ref{input}) and (\ref{approx}) in (\ref{sys}), the system dynamics is converted into an input as well state delayed dynamics as,
\begin{equation}\label{sys new}
\hat{M}(q)\ddot{q}+\bar{N}(q,\dot{q},\ddot{q},\ddot{q}_h)=\tau_h.
\end{equation}
where $\bar{N}=(M-\hat{M})\ddot{q}+\hat{M}_h\ddot{q}_h-\hat{M}u+ N$.


Let, $q^d(t)$ be the desired trajectory to be tracked and $e_1(t)=q(t)-q^d(t)$ is the tracking error. The auxiliary control input $u$ is defined in the following way,
\begin{equation}\label{aux}
u(t)=\ddot{q}^d(t)-K_2\dot{e}_1(t)-K_1e_1(t),
\end{equation}
where, $K_1$ and $K_2$ are two positive definite matrices with appropriate dimensions. Putting (\ref{aux}) and (\ref{input}) in (\ref{sys new}), following error dynamics is obtained,
\begin{equation}\label{error dyn delayed 2}
\ddot{e}_1=-K_2\dot{e}_{1h}-K_1e_{1h}+\sigma_1,
\end{equation}
where, $\sigma_1=(\hat{M}^{-1}\hat{M}_h-I)u_h+\hat{M}^{-1}(\hat{N}_h-\bar{N})+\ddot{q}^d_h-\ddot{q}^d$ and can be treated as overall uncertainty. Further, (\ref{error dyn delayed 2}) can be written in state space form as,
\begin{equation}\label{new err dyn 1}
\dot{e}=A_1e+B_1e_h+B\sigma_1,
\end{equation}
where, $e=\begin{bmatrix}
e_1\\
\dot{e}_1
\end{bmatrix}$, $ A_1=\begin{bmatrix}
0 & I \\
0 & 0
\end{bmatrix} 
, 
B_1=\begin{bmatrix}
0 & 0\\
-K_1 & -K_2
\end{bmatrix}$, $B=\begin{bmatrix}
0\\
I
\end{bmatrix}$. Noting that, $e(t-h)=e(t)-\int\limits_{-h}^0 \dot{e}(t+\theta)\mathrm{d}\theta$, where the derivative inside the integral is with respect to $\theta$, the error dynamics (\ref{new err dyn 1}) is modified as,
\begin{equation}\label{error dyn delayed}
\dot{e}(t)=Ae(t)-B_1\int\limits_{-h}^0 \dot{e}(t+\theta)\mathrm{d}\theta+B\sigma_1,
\end{equation}

where, $A=A_1+B_1$. It is assumed that the choice controller gains $K_1$ and $K_2$ makes $A$ Hurwitz which is always possible. Also, it is assumed that the unknown uncertainties are bounded. In this paper, a new stability criterion, based on the Lyapunov-Krasvoskii method, is presented through Theorem 1 which addresses the issues defined in Problem 1. 
%
\begin{theorem}
The system (\ref{sys new}) employing the control input (\ref{input}), having auxiliary control input (\ref{aux}) is UUB if the controller gains and delay time is selected such that the following condition holds:
\begin{equation}\label{delay value}
\Psi=\begin{bmatrix}
Q-E-(1+\xi)\frac{h^2}{\beta}D & 0\\ 
0 & (\xi-1)\frac{h^2}{\beta}D
\end{bmatrix}>0
\end{equation}

where, $E=\beta PB_1( A_1D^{-1}A_1^T+B_1D^{-1}B_1^{T}+D^{-1} )B_1^{T}P$, $\xi>1$ and $\beta>0$ are scalar, and $P>0$ is the solution of the Lyapunov equation $A^TP+PA=-Q$ for some $Q>0$. 
\end{theorem}
\begin{proof}
Let us consider the following Lyapunov function:
\begin{equation}\label{lyapunov}
V(e)=V_1(e)+V_2(e)+V_3(e)+V_4(e),
\end{equation} 
where,
\begin{align}
V_1(e)&=e^TPe\\
V_2(e)&=\frac{h}{\beta}\int_{-h}^{0}\int_{t+\theta}^{t}e^T(\psi )De(\psi )d\psi d\theta\\
V_3(e)&=\frac{h}{\beta}\int_{-h}^{0}\int_{t+\theta}^{t}e^T(\psi-h )De(\psi-h )d\psi d\theta\\
V_4(e)&=\xi \frac{h^2}{\beta}\int_{t-h}^{t}e^T(\psi)De(\psi)d\psi
\end{align}
Using (\ref{error dyn delayed}), the time derivative  of $V_1(e)$ yields,
\begin{equation}\label{lya_dot for time delay}
\dot{V}_1(e)=-e^TQe-2e^{T}PB_1\int_{-h}^{0}\dot{e}(t+\theta)d\theta+2\hat{s}^T\sigma_1
\end{equation}
where, $\hat{s}=B^TPe$. Again using (\ref{new err dyn 1}),
\begin{align}
-2e^{T}PB_1\int_{-h}^{0}\dot{e} & (t+\theta)d\theta = - 2e^{T}PB_1\int_{-h}^{0}[ A_1e(t+\theta ) +  B_1e(t-h+\theta )+B\sigma _1(t+\theta ) ]d\theta,\label{relation 1} 
\end{align}
For any two non zero vectors $z_1$ and $z_2$, there exists a scalar $\beta>0$ and matrix $D>0$ such that the following inequality holds,
\begin{equation}\label{ineq 2}
\pm 2z_1^{T}z_2\leq \beta z_1^{T}D^{-1}z_1+(1/\beta )z_2^{T}Dz_2.
\end{equation}
Again, using Jensen's inequality the following inequality holds \cite{Ref:31},
\begin{align}
\int_{-h}^{0}e^{T}(\psi)De(\psi)d\psi \geq \frac{1}{h}\int_{-h}^{0}e^{T}(\psi)d\psi D \int_{-h}^{0}e(\psi)d\psi.\label{lemma 1}
\end{align}
Applying (\ref{ineq 2}) and (\ref{lemma 1}) to (\ref{relation 1}) the following inequalities are obtained,
\begin{align}
&- 2e^{T}PB_1 A_1\int_{-h}^{0} e(t+\theta)d\theta \leq \beta e^{T}PB_1A_1D^{-1}A_1^{T}B_1^{T}Pe +\frac{1}{\beta}\int_{-h}^{0}e^{T}(t+\theta )d\theta D \int_{-h}^{0} e(t+\theta ) d\theta\nonumber\\
&\qquad \qquad \qquad \qquad \qquad \qquad \leq \beta e^{T}[PB_1A_1D^{-1}A_1^{T}B_1^{T}P]e +\frac{h}{\beta}\int_{-h}^{0}e^{T}(t+\theta ) D e(t+\theta ) d\theta \label{cond1}\\
&- 2e^{T}PB_1B_1\int_{-h}^{0}   e(t-h+\theta)d\theta \leq \beta e^{T}PB_1B_1D^{-1} B_1^{T}B_1^{T}Pe+\frac{1}{\beta}\int_{-h}^{0}e^{T}(t-h+\theta )d\theta D \int_{-h}^{0} e(t-h+\theta )d\theta  \nonumber\\
& \qquad \qquad \qquad \qquad \qquad \qquad \leq \beta e^{T}[PB_1B_1D^{-1}B_1^{T}B_1^{T}P]e+\frac{h}{\beta}\int_{-h}^{0}e^{T}(t-h+\theta ) D e(t-h+\theta )d\theta \label{cond2} \\
&- 2e^{T}PB_1\int_{-h}^{0} [ B \sigma _1(t+\theta)]d\theta \leq \beta e^{T}PB_1D^{-1}B_1^{T}Pe+ \frac{1}{\beta}\int_{-h}^{0}(B \sigma _1(t+\theta))^{T}d\theta D \int_{-h}^{0} B\sigma _1(t+\theta)]d\theta   \nonumber\\ 
& \qquad \qquad \qquad \qquad \qquad \qquad\leq \beta e^{T}\left [  PB_1D^{-1}B_1^{T}P \right ]e+\frac{h}{\beta}\int_{-h}^{0}(B \sigma _1(t+\theta))^{T} D \sigma _1(t+\theta)  d\theta \label{cond3}
\end{align}
Since $D>0$, we can write $D=\bar{D}^T\bar{D}$ for some $\bar{D}>0$. Then, assuming the uncertainties to be square integrable within the delay, let there exists a scalar $\Gamma_1>0$ such that the following inequality holds:
\begin{align}\label{cond41}
\frac{h}{\beta }\left \| \int_{-h}^{0}\left [ (B \sigma _1(t+\theta))^{T}\bar{D}^T\bar{D}B \sigma _1(t+\theta) \right ]d\theta  \right \|\leq \Gamma_1. 
\end{align}

Again,
\begin{align}
\dot{V}_2(e)&=\frac{h^2}{\beta}e^TDe-\frac{h}{\beta}\int_{-h}^{0}e^T(t+\theta) D e(t+\theta)d\theta \label{dot v2 new}\\ 
\dot{V}_3(e)&=\frac{h^2}{\beta}e^T_h D e_h-\frac{h}{\beta}\int_{-h}^{0}e^T(t-h+\theta) D e(t-h+\theta)d\theta \label{dot v3 new}\\  
\dot{V}_4(e)&=\xi \frac{h^2}{\beta}(e^TDe-e^T_hDe_h)\label{dot v4 new}
\end{align}
Substituting (\ref{cond1})-(\ref{cond41}) into (\ref{lya_dot for time delay}) and adding it with (\ref{dot v2 new})-(\ref{dot v4 new}) yields,
\begin{equation}
\dot{V}(e)\leq -\bar{e}^T \Psi \bar{e}+\Gamma_1+2\hat{s}^T\sigma_1,
\end{equation}
where, $\bar{e}=\begin{bmatrix}
e^T & e^T_h
\end{bmatrix}^T$. Let controller gains $K_1, K_2$ and delay time $h$ are selected to make $\Psi>0$.
One can find a positive scalar $\iota$ such that $|| \hat{s} || \leq \iota || || \bar{e}||$. Then, $\dot{V}(e)<0$ would be established if $\lambda _{min}(\Psi)|| \bar{e}||^{2} >  \Gamma_1 +2 \iota || \sigma_1 ||   || \bar{e}||$.
Thus (\ref{sys new}) would be UUB with the ultimate bound,
\begin{equation}\label{error bound}
||\bar{e}||=\gamma_1+\sqrt{\frac{\Gamma_1}{\lambda _{min}(\Psi)}+\gamma_1^2}=\varpi_0.
\end{equation}
where, $\gamma_1=\frac{\iota || \sigma_1 ||}{\lambda _{min}(\Psi)}$.
Let $\Xi$ denote the smallest level surface of $V$ containing the ball $B_{\varpi_0}$ with radius $\varpi_0$ centred at $\bar{e}=0$. For initial time $t_0$, if $\bar{e}(t_0)\in \Xi$ then the solution remains in $\Xi$. If  $\bar{e}(t_0)\notin \Xi$ then $V$ decreases as long as $\bar{e}(t)\notin \Xi$. The time required to reach $\varpi_0$ is zero when $\bar{e}(t_0)\in \Xi$, otherwise, while $\bar{e}(t_0)\notin \Xi$ the finite time $t_{r0}$ to reach $\varpi_0$, for some $c_0>0$, is given by \cite{Ref:24},
\begin{equation*}
t_{r0}-t_0\leq (||\bar{e}(t_0) ||-\varpi_0)/c_0 \quad \text{where} \quad \dot{V}(t)\leq-c_0
\end{equation*}
\end{proof}
\begin{remark}:
Since $E$ depends on the controller gains, (\ref{delay value}) provides a selection criterion for the choice of delay $h$ for given controller gains and $Q$. This design issue was previously unaddressed in the literature. Moreover, the approximation error $(\hat{N}-N)$, as in (\ref{approx}), would reduce for small values of $h$. However, $h$ cannot be selected smaller than the sampling interval because, the input output data is only available at sampling intervals. So, the lowest possible selection of $h$ is the sampling interval. Again, choice of sampling interval is governed by the corresponding hardware response time, computation time etc. Hence, the proposed stability approach provides a necessary step for the selection of sampling interval for given controller gains or vice-versa.
\end{remark}
\subsection{Filtered Time-Delayed Control (F-TDC)}\label{2.2}

\par It can be noticed from (\ref{approx}) and (\ref{aux}) that state derivatives are necessary to compute the control law of TDC. However, in many circumstances, only $q$ is available amongst $q, \dot{q}, \ddot{q}$. Under this scenario, a new control strategy F-TDC is proposed, which estimates the state derivatives from the state information of past instances \cite{Ref:30}. Before proposing the control structure of F-TDC, the following two Lemmas are stated which are instrumental for formulation as well as stability analysis of F-TDC.

\begin{lemma}
[\cite{Ref:30}]: For time $t \geq \varsigma$, the $j$-th order time derivative of the $\Lambda$-th degree polynomial $q$ in (\ref{sys new}) can be computed in the following way,
\begin{equation}
\hat{q}^{(j)}(t)=\int_{-\varsigma}^{0}\Omega _{j}(\varsigma ,\psi)q(t+\psi)d\psi\label{lemma 2}
\end{equation}
where, 
$\varsigma>0$ is a prespecified scalar and 
\begin{align}
\Omega _{j}(\varsigma ,\psi) =\frac{(\Lambda+1+j)!}{\varsigma ^{(j+1)}j!(\Lambda-j)!}& \sum_{k=0}^{\Lambda}\frac{(-1)^{k}(\Lambda+1+k)!}{(j+k+1)(\Lambda-k)!(k!)^{2}} \left ( \frac{-\psi}{\varsigma} \right )^{k}.\label{omega}
\end{align}
\end{lemma}
\begin{lemma}
For any non zero vector $\vartheta(\psi)$, constant matrix $F>0$ the following relation holds,
\begin{align}
&\int_{-h}^{0}\int_{-\varsigma}^{0}\vartheta ^T(\psi) F \vartheta(\psi) d\psi d\theta \geq \frac{1}{h \varsigma}\left \{  \int_{-h}^{0}\int_{-\varsigma}^{0}\vartheta ^T(\psi)d\psi d\theta\right \} F \left \{  \int_{-h}^{0}\int_{-\varsigma}^{0}\vartheta(\psi)d\psi d\theta\right \}\label{lemma 3}
\end{align}
\end{lemma}
The structure of F-TDC is similar to (\ref{input}), except, the auxiliary control input $u$ and $\hat{N}$ in (\ref{input}) selected in the following way,
\begin{align}
u(t)&=\ddot{q}^d(t)-K_1e_1(t)-K_2\dot{\hat{e}}_1(t)\label{u hat}\\
\hat{N}(t)& \cong N_{h}=u_{h}-\hat{M}_{h}\ddot{\hat{q}}_{h}, \label{h hat new}
\end{align}
where, $\dot{\hat{e}}_1=\dot{\hat{q}}-\dot{q}^d$. $\dot{\hat{q}}$ and $\ddot{\hat{q}}$ are evaluated from (\ref{lemma 2}) and (\ref{omega}).
The stability of the system (\ref{sys new}) employing F-TDC is derived in the sense of  Uniformly Ultimately Bounded (UUB) notion as stated in Theorem 2.

\begin{theorem}
The system (\ref{sys new}) employing the control input (\ref{input}), having the auxiliary input (\ref{u hat}) and (\ref{h hat new}) is UUB if $K_1, K_2, h$ and $\varsigma$ are selected such that the following condition holds:
%
\begin{align}
&\begin{bmatrix}
Q-\bar{E}-(1+\xi)\frac{h^2}{\beta}D& P \breve{B} & P\bar{B}\\ 
 \breve{B}^T P & (\xi-1)\frac{h^2}{\beta}D-\bar{F} &0\\
\bar{B}^TP & 0 & L
\end{bmatrix} \nonumber \\
& \qquad \qquad \qquad \qquad \qquad =\Theta>0 \label{delay value 1}
\end{align}

where, $\bar{E}=\beta P {B}_1( {A}_1D^{-1}{A}_1^T+{B}_1D^{-1}{B}_1^{T}+D^{-1}+\bar{B}D^{-1}\bar{B}^{T} ){B}_1^{T}P$, $\bar{F}=(\frac{h^2}{\beta}D+L) \varsigma \int_{-\varsigma}^{0}A_d^2(\psi)d\psi$, $L>0$, $A_d(\psi)=\Omega _{1}(\varsigma ,\psi)$, $\bar{B}=B\begin{bmatrix}
 K_2 & 0
\end{bmatrix}$, $\breve{B}=B\begin{bmatrix}
 0 & K_2
\end{bmatrix}$.
\end{theorem}
\begin{proof}
The proof is provided in Appendix B.
\end{proof}

\subsection{Adaptive-Robust Control: Related Work}
It can be observed that TDE degrades tracking performance of both the TDC and F-TDC in the face of uncertainties. The control methods that attempt to counter uncertainties, as reported in \cite{Ref:2}, \cite{Ref:10}, \cite{Ref:12}, \cite{Ref:13}-\cite{Ref:19}, requires predefined bound on the uncertainties which is not always possible in practical circumstances. To circumvent this situation Adaptive Sliding Mode Control (ASMC) was proposed in \cite{Ref:21}-\cite{Ref:22}. The control input of ASMC is given by,
\begin{equation}\label{input asmc}
\tau=\Sigma_n^{-1}(-\kappa_n+\Delta u_s),
\end{equation}
where, $\Sigma_n$ and $\kappa_n$ is the nominal values of $\Sigma$ and $\kappa$, and $\Delta u_s$ is the switching control input. For a choice of sliding surface $\bar{s}$, $\Sigma$ and $\kappa$ is defined as follows:
\begin{equation}\label{s asmc}
\dot{\bar{s}}=\Sigma+\kappa\Delta u_s,
\end{equation}
The switching control $\Delta u_s$ is calculated as
\begin{equation}\label{delta u asmc}
 \Delta u_s=-\hat{c}\frac{\bar{s}}{||\bar{s}||}
\end{equation}
\begin{equation}\label{ASMC}
 \dot{\hat{c}} =
  \begin{cases}
   \bar{c}|| \bar{s}|| sgn(|| \bar{s}||-\rho)      & \quad \hat{c}>\gamma,\\
    \qquad \qquad \gamma        & \quad \hat{c}\leq\gamma,\\
  \end{cases}
\end{equation}
where, $\hat{c}$ is the switching gain, $\bar{c}>0$ is a scalar adaptive gain, $\rho>0$ is a threshold value, $\gamma >0$ is small scalar to always keep $\hat{c}$ positive. Evaluation of $\rho$ can be done in two ways \cite{Ref:21}:
\begin{gather}
\rho=\varrho \quad \text{or},\label{ASMC 1}\\
 \rho(t)=4\hat{c}(t)t_s,\label{ASMC 2}
\end{gather}
where, $\varrho>0$ is a scalar, $t_s$ is the sampling interval. The choice (\ref{ASMC 1}) requires predefined bound of uncertainties. It can be noted from (\ref{ASMC}) that even if $|| \bar{s}||$ decreases (resp. increases), unless it falls below (resp. goes above) $\rho$  switching gain does not decrease (resp. increase). This causes overestimation (resp. underestimation) of switching gain and controller accuracy is compromised. Again, improper and low choice of $\rho$ may lead to very high switching gain and consequent chattering. On the other hand, method (\ref{ASMC 2}) assumes that the nominal value of the uncertainties are always greater than the perturbations. This assumption may not hold due to the effect of unmodelled dynamics and thus,  necessitates rigorous nominal modelling of the uncertainties in $N$ to design the control law. Either of the two situations, i.e. bound estimation or uncertainty modelling, is not always feasible in practical circumstances and consequently compromises the adaptive nature of the controller.
\subsection{Time-Delayed Adaptive Robust Control}\label{sec: 2.3}
Considering the limitations of the existing controllers that aim at negotiating the uncertainties, as discussed earlier, a novel adaptive-robust control law, named Time-Delayed Adaptive Robust Control (TARC) is proposed in this endeavour, which neither requires the nominal model nor any predefined bound of the uncertainties as well as eliminates the overestimation-underestimation problem of switching gain. The structure of the control input of TARC is similar to (\ref{input}) and $\hat{N}$ is also evaluated according to (\ref{h hat new}). However, the auxiliary control input $u$ is selected as below,
\begin{equation}\label{tarc input}
u = \hat{u}+\Delta u.
\end{equation}
$\hat{u}$ is the nominal control input and selected as similar to (\ref{u hat}). $\Delta u$ is the switching control law which is responsible for negotiating the TDE and it is defined as below,
\begin{align}
\Delta u=
  \begin{cases}
    -\alpha\hat{c}(e,t)\frac{s}{\parallel s \parallel}       & \quad \text{if } \parallel s \parallel\geq \epsilon,\\
    -\alpha\hat{c}(e,t)\frac{s}{\epsilon}        & \quad \text{if } \parallel s \parallel< \epsilon,\\
  \end{cases}\label{delta u}
\end{align}
where, $s=B^TP\begin{bmatrix}
e_1 & \dot{\hat{e}}_1
\end{bmatrix}^T$ and $\epsilon>0$ is a small scalar.
The following novel adaptive control law for evaluation of $\hat{c}$ is proposed:
\begin{equation}\label{ATRC}
 \dot{\hat{c}} =
  \begin{cases}
    \quad || s ||      & \quad \hat{c}>\gamma,f(e)>0\\
    -|| s ||              & \quad \hat{c}>\gamma,f(e) \leq 0 \\
    \quad  \gamma        & \quad \hat{c}\leq\gamma,\\
  \end{cases}
\end{equation}
where, $\alpha>0$ is a scalar adaptive gain  and $\epsilon>0$ represents a small scalar, $f(e)$ is a suitable function of error defined by the designer. Here, it is selected as $ f(e)=|| s(t)||-|| s_h  ||$. According to the adaptive law (\ref{ATRC}) and present choice of $f(e)$, $\hat{c}$ increases (resp. decreases) whenever error trajectories move away from (resp. close to ) $|| s ||=0$  The advantages of the proposed TARC can be summarized as follows:
\begin{itemize}
  \item TARC reduces complex system modelling effort as only the knowledge of $\hat{M}$ suffices the controller design since $N$ along with the uncertainties is approximated using the time-delayed logic as in (\ref{h hat new}). This in turn reduces the tedious modelling effort of complex nonlinear systems.
  \item Evaluation of switching gain does not require either of the nominal model or predefined bound of the uncertainties and also removes the overestimation-underestimation problem.
  \item State derivatives are not required to compute the control law explicitly, as they are evaluated from the past state information using (\ref{lemma 2}) and (\ref{omega}).
\end{itemize}

The stability of the system (\ref{sys new}) employing TARC is analysed in the sense of UUB as stated in Theorem 3.
\begin{assumption}\label{assum 3}
Let, $|| \sigma_1|| \leq c$. Here, $c$ is an unknown scalar quantity. Knowledge of $c$, however, is only required for stability analysis but not to compute control law. 
\end{assumption}
\begin{theorem}
The system (\ref{sys new}) employing (\ref{input}), (\ref{tarc input}) and having the adaptive law (\ref{ATRC}) is UUB, provided the selection of $K_1, K_2, h$ and $\varsigma$ holds condition (\ref{delay value 1}).
\end{theorem}
\begin{proof}:
Let us define the Lyapunov functional as,
\begin{align}
V_r(e)&=V_f(e)+(\hat{c}-c)^2,\label{tarc lya}
\end{align}
where, $V_f(e)$ is defined in (\ref{ftdc lya}). Again, putting (\ref{tarc input}) in (\ref{sys new}) the  the error dynamics becomes,
\begin{equation}\label{error dyn state new}
\dot{e}={A}_1e+{B}_1e_h-\bar{B}\int_{-\varsigma}^{0}A_d(\psi )e(t-h+\psi )d\psi+B\sigma,
\end{equation}
where, $\sigma=\Delta u_h+\sigma_1$. Also following similar steps while proving Theorem 2 (provided in Appendix B) we have,
\begin{equation}\label{vf dot}
\dot{V}_f(e)\leq -e_f^T \Theta e_f + \Gamma +2\hat{s}^T(\Delta u+\sigma_1)+2\hat{s}^T \Upsilon,
\end{equation}
where, $ e_f$ is defined in Appendix B, $\Gamma \geq \frac{h}{\beta }\left \| \int_{-h}^{0}\left [ (B \sigma(t+\theta))^{T}\bar{D}^T\bar{D}B \sigma(t+\theta) \right ]d\theta  \right \|$ is a positive scalar, $\Upsilon=\Delta u_h-\Delta u$. Let us define the following,
\begin{equation}
\hat{s}=s+\Delta s \quad \text{where} \quad \Delta s=B^TP\begin{bmatrix}
0 & (\dot{e}_1-\dot{\hat{e}}_1)^T
\end{bmatrix}^T.
\end{equation}
Evaluating the structure of $s$ and $\Delta s$ one can find two positive scalars $\iota_2, \iota_3$ such that $||{s}||\leq \iota_2 ||e_f||, ||\Delta {s}||\leq \iota_3 ||e_f|| $.
Using (\ref{tarc lya}) the stability analysis for (\ref{sys new}) employing TARC is carried out for the following various cases.
\par \textbf{Case (i):} $f(e)>0,\quad \hat{c}>\gamma,\quad || s ||\geq \epsilon$. \\
Utilizing (\ref{delta u}), (\ref{ATRC}) and (\ref{vf dot}) we have,
\begin{align}
\dot{V}_r(e)&\leq -e_f^T \Theta e_f+\Gamma+2\hat{s}^T(-\alpha \hat{c} \frac{s}{|| s ||}+\sigma_1)+2\hat{s}^T \Upsilon+2(\hat{c}-c)||s||\nonumber\\
& =-e_f^T \Theta e_f+\Gamma-2\alpha\hat{c}\frac{s^Ts}{||s||}-2\alpha\hat{c}\frac{\Delta s^Ts}{||s||}+2\hat{s}^T \sigma_1 +2\hat{s}^T \Upsilon +2(\hat{c}-c)||s||\nonumber\\
& \leq -\lambda_{min}(\Theta)||e_f||^2-2(\alpha-1)\hat{c}||s||+\Gamma+2||s||||\Upsilon|| +2(\alpha\hat{c} +c+||\Upsilon||)||\Delta s||\label{case 1}
\end{align}
So, for $\alpha>1$, $\dot{V}_r(e)<0$ would be established if $\lambda_{min}(\Theta)||e_f||^2>\Gamma+2||s||||\Upsilon||+2(\alpha\hat{c} +c+||\Upsilon||)||\Delta s||$. Thus, using the relation $||{s}||\leq \iota_2 ||e_f||, ||\Delta {s}||\leq \iota_3 ||e_f|| $, the system would be UUB with the following ultimate bound
\begin{equation}\label{error bound 1}
||e_f||=\mu_1+\sqrt{\frac{\Gamma}{\lambda _{min}(\Theta)}+\mu_1^2}=\varpi_1.
\end{equation}
where, $\mu_1=\frac{\iota_2 || \Upsilon ||+\iota_3(\alpha\hat{c} +c+||\Upsilon||)}{\lambda _{min}(\Theta)}$.
\par \textbf{Case (ii):} $f(e) \leq 0,\quad \hat{c}>\gamma,\quad \left \| s \right \|\geq \epsilon$. \\
Again, utilizing (\ref{ATRC}) for Case (ii),
\begin{align}
\dot{V}_r(e)&\leq -e_f^T \Theta e_f+\Gamma+2\hat{s}^T(-\alpha \hat{c} \frac{s}{|| s ||}+\sigma_1)+2\hat{s}^T \Upsilon -2(\hat{c}-c)||s||\nonumber\\
& \leq -\lambda_{min}(\Theta)||e_f||^2+( 4c-2(\alpha+1)\hat{c}+2||\Upsilon|| )||s|| +\Gamma +2(\alpha\hat{c} +c+||\Upsilon||)||\Delta s||.\label{case 2}
\end{align}
$\dot{V}_r(e)<0$ would be achieved if $\lambda_{min}(\Theta)||e_f||^2> \Gamma +( 4c-2(\alpha+1)\hat{c}+2||\Upsilon|| )||s||+2(\alpha\hat{c} +c+||\Upsilon||)||\Delta s||$ and system would be UUB having following ultimate bound,
\begin{equation}\label{error bound 2}
||e_f||=\mu_2+\sqrt{\frac{\Gamma}{\lambda _{min}(\Theta)}+\mu_2^2}=\varpi_2.
\end{equation}
where, $\mu_2=\frac{\iota_2 ( 2c-(\alpha+1)\hat{c}+||\Upsilon|| )+\iota_3(\alpha\hat{c} +c+||\Upsilon||)}{\lambda _{min}(\Theta)}$.

\par \textbf{Case (iii):} $\hat{c}\leq \gamma,\quad \left \| s \right \|\geq \epsilon$.\\
Since $\hat{c}\leq \gamma$ we have $(\hat{c}-c )\gamma \leq \gamma^2-c\gamma \leq \gamma^2$. Using the adaptive law (\ref{ATRC}), for Case (iii) we have,
\begin{align}
\dot{V}_r(e)& \leq -e_f^T \Theta e_f+\Gamma+2\hat{s}^T(-\alpha \hat{c} \frac{s}{|| s ||}+\sigma_1)+2(\hat{c}-c)\gamma +2\hat{s}^T \Upsilon\nonumber\\
& \leq -\lambda_{min}(\Theta)||e_f||^2+\Gamma+2( c-\alpha\hat{c}+||\Upsilon||)||s|| +2(\alpha\hat{c} +c+||\Upsilon||)||\Delta s||+2\gamma^2.\label{case 3}
\end{align}
Similarly, as argued earlier the system would be UUB with the following ultimate bound,
\begin{equation}\label{error bound 3}
||e_f||=\mu_3+\sqrt{\frac{(\Gamma+2\gamma^2)}{\lambda _{min}(\Theta)}+\mu_3^2}=\varpi_3.
\end{equation}
where, $\mu_3=\frac{\iota_2 ( c-\alpha\hat{c}+||\Upsilon|| )+\iota_3(\alpha\hat{c} +c+||\Upsilon||)}{\lambda _{min}(\Theta)}$.

\end{proof}
\begin{remark}
The performance of TARC can be characterized by the various error bounds under various conditions. It can be noticed that low value of $h$ and high value of $\alpha$ would result in better accuracy. However, too large $\alpha$ may result in high control input. Also, one may choose different values of $\alpha$ for $|| s || > || s_h||$ and $|| s || \leq || s_h||$. Moreover, it is to be noticed that the stability notion of TARC is invariant to the choice of $f(e)$ and thus provides the designer the flexibility to select a suitable $f(e)$ according to the application requirement.   
\end{remark}

\section{Conclusion}
Selection of the controller gain and sampling interval is crucial for the performance of TDC and this design issue is addressed in this paper through a new stability approach. A bound on the delay is derived to select a suitable sampling interval. A new control approach, F-TDC is devised where the state derivatives are estimated from the previous state information. Moreover, a novel adaptive-robust control law, TARC has been proposed for a class of uncertain nonlinear systems subjected to unknown uncertainties. The proposed controller approximates unknown dynamics through time-delayed law and negotiates the approximation error, that surfaces due to the time-delayed approximation of uncertainties and state derivatives, by switching logic. The adaptive law eliminates the overestimation-underestimation problem for online evaluation of switching gain without any prior knowledge of uncertainties. Experimentation with a WMR shows improved path tracking performance of TARC compared to TDC and conventional ASMC. The proposed framework can also be extended for other systems such as Autonomous Underwater Vehicle, Unmanned Aerial Vehicle, Robotic manipulator etc. 

\appendices
\section{Proof of Theorem 2} 
Let us define the Lyapunov functional as,
\begin{align}
V_f(e)&=V(e)+V_{f1}(e)+V_{f2}(e)+V_{f3}(e),\label{ftdc lya}\\
V_{f1}(e)&=\frac{h \varsigma}{\beta}\int_{-h}^{0}\int_{-\varsigma}^{0}\int_{t-h+\psi}^{t-h}e^T(\eta +\theta)D\times \nonumber\\
& \qquad \qquad \qquad \qquad \qquad \times A_d^2 (\psi)e(\eta+\theta)d\eta d\psi d\theta\nonumber\\
V_{f2}(e)&=\frac{h \varsigma}{\beta}\int_{-h}^{0}\int_{-\varsigma}^{0}\int_{t+\theta}^{t}e^T(\eta-h)D \times \nonumber\\
& \qquad \qquad \qquad \qquad \qquad \times A_d^2(\psi)e(\eta-h)d\eta d\psi d\theta\nonumber\\
V_{f3}(e)&=\varsigma \int_{-\varsigma}^{0}\int_{t+\psi}^{t}e^T(\eta-h)(\psi)RA_d^2(\psi)e(\eta-h)d\eta d\psi.\nonumber
\end{align}
\ifCLASSOPTIONcaptionsoff
  \newpage
\fi


\begin{thebibliography}{1}

\bibitem{Ref:1}
M. Kristic, I. Kanellakopoulos and P. V. Kokotovic, \emph{Nonlinear and Adaptive Control Design}, Wiley, New York, 1995.
\bibitem{Ref:2}
X. Liu, H. Su, B. Yao and J.Chu, Adaptive robust control of a class of uncertain nonlinear systems with unknown sinusoidal disturbances, \emph{47th  IEEE Conference on  Decision and Control}, pp. 2594-2599, 2008.
\bibitem{Ref:3}
M. J. Corless and G. Leitmann, Continuous state feedback guaranteeing uniform ultimate boundness for uncertain dynamic system, \emph{IEEE Transactions on Automatic Control}, vol. 26(10), pp. 1139-1144, 1981.
\bibitem{Ref:4}
H. Lee and V. I. Utkin, Chattering suppression methods in sliding mode control systems, \emph{Annual Reviews in Control}, vol. 31, pp. 179-188, 2007.
\bibitem{Ref:5}
A. Levant, Higher-order sliding modes, differentiation and output-feedback control, \emph{International Journal of Control}, vol. 76(9/10), pp. 924-941, 2003.
\bibitem{Ref:6-1}
A. G. Ulsoy, Time-delayed control of siso systems for improved stability margins, \emph{ASME Journal of Dynamic Systems, Measurement, and Control}, vol. 137(4), pp. 1-12, 2015.
\bibitem{Ref:6}
T. C. Hsia and L. S. Gao, Robot manipulator control using decentralized linear time-invariant time-delayed joint controllers, \emph{IEEE International Conference on  Robotics and Automation}, pp. 2070-2075, 1990.
\bibitem{Ref:6-2}
K. Youcef-Toumi and O. Ito, A time delay controller for systems with unknown dynamics, \emph{ASME Journal of Dynamic Systems, Measurement, and Control}, vol. 112(1), pp. 133-142, 1990.

\bibitem{Ref:7}
G. R. Cho, P. H. Chang, S. H. Park and M. Jin, Robust tracking under nonlinear friction using time delay control with internal model, \emph{IEEE Transactions on Control System Technology}, vol. 17(6), pp. 1406-1414, 2009.
\bibitem{Ref:8}
D. K.Han and P. H. Chang, Robust tracking of robot manipulator with nonlinear friction using time delay control with gradient estimator, \emph{Journal of Mechanical Science and Technology}, vol. 24(8), pp. 1743-1752, 2010.
\bibitem{Ref:9}
M. Jin, S. H. Kang and P. H. Chang, Robust compliant motion control of robot with nonlinear friction using time-delay estimation, \emph{IEEE Transactions on Industrial Electronics}, vol. 55(1), pp. 258-269, 2008.
\bibitem{Ref:10}
Y. Jin, P. H.  Chang, M. Jin and D. G. Gweon, Stability guaranteed time-delay control of manipulators using nonlinear damping and terminal sliding mode, \emph{IEEE Transactions on Industrial Electronics}, vol. 60(8), pp. 3304-3318, 2013.
\bibitem{Ref:11}
P. H. Chang and S. H. Park, On improving time-delay control under certain hard nonlinearities, \emph{Mechatronics}, vol. 13(4), pp. 393-412, 2003.
\bibitem{Ref:12}
S. Roy, S. Nandy, R. Ray and S. N. Shome, Time delay sliding mode control nonholonomic wheeled mobile robot: experimental validation, \emph{IEEE International Conference on  Robotics and Automation}, pp. 2886-2892, 2014.
\bibitem{Ref:13}
X. Zhu, G. Tao, B. Yao and J. Cao, Adaptive robust posture control of parallel manipulator driven by pneumatic muscles with redundancy, \emph{IEEE/ASME Transactions on Mechatronics}, vol. 13(4), pp. 441-450, 2008.
\bibitem{Ref:14}
X. Zhu, G. Tao, B. Yao and J. Cao, Integrated direct/indirect adaptive robust posture control of parallel manipulator driven by pneumatic muscles, \emph{IEEE Transactions on Control System Technology}, vol. 17(3), pp. 576-588, 2009.
\bibitem{Ref:16}
W. Sun, Z. Zhao and H. Gao, Saturated adaptive robust control for active suspension systems, \emph{IEEE Transactions on Industrial Electronics}, vol. 60(9), pp.  3889-3896, 2013.
\bibitem{Ref:17}
S. Islam, P. X. Liu and A. E. Saddik, Robust control of four rotor unmanned aerial vehicle with disturbance uncertainty, \emph{IEEE Transactions on Industrial Electronics}, vol. 63(3), pp. 1563-1571, 2015.
\bibitem{Ref:18}
Z. Chen, B. Yao and Q. Wang, $\mu$-synthesis based adaptive robust posture control of linear motor driven stages with high frequency dynamics: a case study, \emph{IEEE/ASME Transactions on Mechatronics}, vol. 20(3), pp. 1482-1490, 2015.
\bibitem{Ref:19}
Z. Liu, H. Su and S. Pan, A new adaptive sliding mode control of uncertain nonlinear dynamics, \emph{Asian Journal of Control}, vol. 16(1), pp. 198-208, 2014.
\bibitem{Ref:20}
C. Chen, T. S. Li, Y. Yeh and C. Chang, Design and implementation of an adaptive sliding-mode dynamic controller for wheeled mobile robots, \emph{Mechatronics}, vol. 19, pp. 156-166, 2009.
\bibitem{Ref:21}
F. Plestan, Y. Shtessel, V. Bregeault and A. Poznyak, New methodologies for adaptive sliding mode control, \emph{International Journal of Control}, vol. 83(9), pp. 1907-1919, 2010.
\bibitem{Ref:22}
F. Plestan, Y. Shtessel, V. Bregeault and A. Poznyak, Sliding mode control with gain adaptation - application to an electropneumatic actuator state, \emph{Control Engineering Practice}, vol. 21, pp. 679-688, 2013.
\bibitem{Ref:23}
B. Bandyopadhayay, S. Janardhanan and  S. K. Spurgeon, \emph{Advances in Sliding Mode Control}, Springer-Verlag, New York, 2013.
\bibitem{Ref:24}
G. Leitmann, On the efficiency of nonlinear control in uncertain linear systems, \emph{ASME Journal of Dynamic Systems, Measurement and Control}, vol. 103, pp. 95-102, 1981.
\bibitem{Ref:25}
P. H. Chang and J. W. Lee, A model  reference observer for time-delay control and its application to robot trajectory control, \emph{IEEE Transactions on Control System Technology}, vol. 4(1), pp. 2-10, 1996.
\bibitem{Ref:26}
Y. H. Shin and K. J. Kim, Performance enhancement of pneumatic vibration isolation tables in low frequency range by time delay control, \emph{Journal of Sound and Vibration}, vol. 321, pp. 537-553, 2009.

\bibitem{Ref:27}
A. Kuperman and Q. C. Zhong, Robust control of uncertain nonlinear systems based on an uncertainty and disturbance estimator, \emph{International Journal of Robust and Nonlinear Control}, vol. 21, pp. 79-92, 2011.
\bibitem{Ref:28}
S. E. Talole and S. B. Phadke, Robust input-output linearization using uncertainty and disturbance estimator, \emph{International Journal of Control}, vol. 82(10), pp. 1794-1803, 2009.
\bibitem{Ref:29}
P. V. Suryawanshi, P. D. Shengde, and S. B. Phadke, Robust sliding mode control for a class of nonlinear systems using inertial delay control, \emph{Nonlinear Dynamics}, vol. 78(3), pp. 1921-1932, 2014.
\bibitem{Ref:30}
J. Reger and J. Jouffroy, On algebraic time-derivative estimation and deadbeat state reconstruction, \emph{ Joint 47th  IEEE Conference on  Decision and Control and 28th Chinese Control Conference}, pp. 1740-1745, 2009.
\bibitem{Ref:31}
K. Gu, V. Khartionov and J. Chen, \emph{Stability of time-delay systems}, Birkh\"{a}user, Boston (2003).
\bibitem{Ref:32}
A. Nasiri, S. K. Nguang and A. Swain, Adaptive sliding mode control for a class of MIMO nonlinear systems with uncertainty, \emph{Journal of The Franklin Institute}, vol. 351, pp. 2048-2061, 2014.
\bibitem{Ref:33}
S. Roy, S. Nandy, R. Ray and S. N. Shome, Robust path tracking control of nonholonomic wheeled mobile robot, \emph{International Journal of Control, Automation and Systems}, vol. 13(4), pp. 897-905, 2015.
\bibitem{Ref:34}
J. Kim, H. Joe, S-C Yu, J. S. Lee and M. Kim, Time delay controller design for position control of autonomous underwater vehicle under disturbances, \emph{IEEE Transactions on Industrial Electronics},  DOI 10.1109/TIE.2015.2477270, 2015.
\end{thebibliography}
\end{document}